\documentclass[a4paper,11pt]{article}
\usepackage[latin1]{inputenc}
\usepackage[T1]{fontenc}
\usepackage{amsmath,amsthm,amsfonts,amssymb}
\usepackage{enumerate}
\usepackage[colorlinks=true, allcolors=blue]{hyperref}
\usepackage{color}
\usepackage{xspace}
\usepackage{comment}
\usepackage{graphicx}

\usepackage{booktabs}

\newcommand{\cC}{{\mathcal C}}

\newcommand{\Z}{\mathbb{Z}}

\newcommand{\R}{\mathbb{R}}

\newcommand{\E}{\mathbb{E}}
\newcommand{\pr}{\mathbb{P}}

\newcommand{\var}{\operatorname{Var}}

\newcommand{\Ber}{\operatorname{Ber}}
\newcommand{\Bin}{\operatorname{Bin}}
\newcommand{\Poi}{\operatorname{Poi}}

\newcommand{\Par}{\operatorname{Par}}
\newcommand{\MPoi}{\operatorname{MPoi}}
\newcommand{\LogN}{\operatorname{LNor}}

\newcommand{\weq}{\ = \ }

\newcommand{\wle}{\ \le \ }
\newcommand{\wge}{\ \ge \ }

\newcommand{\lest}{\le_{\rm{st}}}

\newcommand{\lecx}{\le_{\rm{cx}}}
\newcommand{\gecx}{\ge_{\rm{cx}}}
\newcommand{\lecv}{\le_{\rm{cv}}}

\newcommand{\leicx}{\le_{\rm{icx}}}

\newcommand{\leicv}{\le_{\rm{icv}}}

\newcommand{\leLt}{\le_{\rm{Lt}}}

\newcommand{\abs}[1]{\lvert #1 \rvert}

\newcommand{\dto}{\xrightarrow{d}}

\newcommand{\Wto}{\xrightarrow{W}}

\theoremstyle{plain}
\newtheorem{theorem}{Theorem}[section]
\newtheorem{lemma}[theorem]{Lemma}
\newtheorem{proposition}[theorem]{Proposition}

\newtheorem*{theorem*}{Theorem}
\newtheorem*{lemma*}{Lemma}
\newtheorem*{proposition*}{Proposition}
\newtheorem*{conjecture*}{Conjecture}
\newtheorem{fact*}{Fact}

\theoremstyle{definition}

\newtheorem{example}[theorem]{Example}
\newtheorem{remark}[theorem]{Remark}
\newtheorem*{definition*}{Definition}
\newtheorem*{question*}{Question}
\newtheorem*{example*}{Example}
\newtheorem*{remark*}{Remark}

\numberwithin{equation}{section}

\newcommand{\dtv}{d_{\rm{tv}}}

\newcommand{\Cmax}{\abs{\cC_{\rm{max}}}}
\newcommand{\alphacr}{\alpha_{\rm cr}}
\newcommand{\lambdacr}{\lambda_{\rm cr}}

\newcommand{\zetas}{\zeta_{\rm CM}}

\newif\ifshowb
\newif\ifshowr

\showbfalse
\showrfalse

\ifshowb
  
\else
  \excludecomment{bcomm}
\fi

\ifshowr
  
\else
  \excludecomment{rcomm}
\fi

\newcommand{\cred}[1]{#1} 

\newcommand{\fracpadding}{}
\newcommand{\setfracpadding}[1][2pt]{%
  \sbox0{$\frac{1}{2}$}%
  \dimen0=\ht0 \advance\dimen0 #1\relax
  \dimen2=\dp0 \advance\dimen2 #1\relax
  \edef\fracpadding{\vrule width 0pt height \the\dimen0 depth \the\dimen2\relax}%
}

\title{The impact of degree variability on connectivity properties of large networks\thanks{A preliminary version of this work was presented at the 12th Workshop on Algorithms and Models for the Web Graph (WAW '15), Eindhoven, Netherlands, December 2015. }}

\author{
Lasse Leskel\"a\thanks{
 Aalto University School of Science,
 Department of Mathematics and Systems Analysis,
 Otakaari 1 Espoo, Finland. 
 \ \href{http://math.aalto.fi/en/research/stochastics/}{\nolinkurl{math.aalto.fi/en/research/stochastics/}}
 }
\and
Hoa Ngo\footnotemark[2]
}

\date{13 January 2017}

\begin{document}
\maketitle
\begin{abstract}
The goal of this work is to study how increased variability in the degree distribution impacts the global connectivity properties of a large network. We approach this question by modeling the network as a uniform random graph with a given degree sequence. We analyze the effect of the degree variability on the approximate size of the largest connected component using stochastic ordering techniques. A counterexample shows that a higher degree variability may lead to a larger connected component, contrary to basic intuition about branching processes. When certain extremal cases are ruled out, the higher degree variability is shown to decrease the limiting approximate size of the largest connected component.
\end{abstract}


\section{Introduction}

Digital communication networks and online social media have dramatically increased the spread of information in our society. As a result, the global connectivity structure of communication between people appears to be better modeled as a dimension-free unstructured graph instead of a geometrical graph based on a two-dimensional grid, and the spread of messages over an online network can be modeled as an epidemic on a large random graph. When the nodes of the network spread the epidemic independently of each other, the final outcome of the epidemic, or the ultimate set of nodes that receive a message, corresponds to the connected component of the initial root node in a randomly thinned version of the original communication graph called the epidemic generated graph \cite{Ball_Sirl_Trapman_2014}. This is why the sizes of connected components are important in studying information dynamics in unstructured networks.

A characterizing statistical feature of many communication networks is the high variability among node degrees, which is manifested by observed approximate power-law shapes in empirical measurements. The simplest mathematical model that allows to capture the degree variability is the so-called configuration model which is defined as follows. Fix a set of nodes labeled using $[n] = \{1,2,\dots,n\}$ and a sequence of nonnegative integers
$
 d_n = \{d_n(1),\dots,d_n(n)\}
$ 
such that $\ell_n = \sum_{i=1}^n d_n(i)$ is even. Each node $i$ gets assigned $d_n(i)$ half-links, or stubs, and then we select a uniform random matching among the set of all half-links. A matched pair of half-links will form a link, and we denote by $X_{i,j}$ the number of links with one half-link assigned to $i$ and the other half-link assigned to $j$. The resulting random matrix $(X_{i,j})$ constitutes a random undirected multigraph on the node set $[n]$. This multigraph is  called the \emph{configuration model} generated by the degree sequence $d_n$. The multigraph is called simple if it contains no loops ($X_{i,i}=0$ for all $i$) and no parallel links ($X_{i,j} \le 1$ for all $i,j$). The distribution of the multigraph conditional on being simple is the same as the distribution of the uniform random graph in the space of graphs on $[n]$ with degree sequence $d_n$ \cite[Prop.~7.13]{VanDerHofstad_2014_Vol1}.

A tractable mathematical way to analyze large random graphs is to let the size of the graph grow to infinity and approximate the empirical degree distribution of the random graph
\[
 p_n(k)  = \frac{1}{n} \sum_{i=1}^n 1(d_n(i) = k)
\]
using a limiting probability distribution $p$ on the infinite set of nonnegative integers $\Z_+$. One of the key results in the theory of random graphs is the following, first derived by Molloy and Reed \cite{Molloy_Reed_1995,Molloy_Reed_1998} and strengthened by Janson and {\L}uczak \cite{Janson_Luczak_2009}.
%
%
Assume that the collection of degree sequences $(d_n)$ is such that the corresponding empirical degree distributions satisfy
\begin{equation}
\label{eq:CMRegular}
\begin{aligned}
 & p_n(k) \to p(k) \quad \text{for all $k \ge 0$}, \\
 & \quad \quad \sup_{n} m_2(p_n) < \infty, 
\end{aligned}
\end{equation}
and that $p(2) < 1$ and $0 < m_1(p) < \infty$, where $m_i(p) = \sum_k k^i p(k)$ denotes the $i$th moment $p$.
Then \cite[Thm 2.3, Rem~2.7]{Janson_Luczak_2009} the size of the largest connected component $\Cmax$ in the configuration model multigraph (and in the associated uniform random graph) converges according to
\begin{equation}
 \label{eq:MaxComponent}
 n^{-1} \Cmax \to \zetas(p) \qquad \text{(in probability)},
\end{equation}
where the constant $\zetas(p) \in [0,1]$ is uniquely characterized by $p$ and satisfies $\zetas(p) > 0$ if and only if $m_2(p) > 2 m_1(p)$. The above fundamental result is important because it has been extended to models of wide generality (e.g.\ \cite{Bollobas_Janson_Riordan_2011}).

Most earlier mathematical studies (and extensions) have focused on establishing the phase transition (showing that there is a critical phenomenon related to whether or not $\zetas(p) > 0$), and studying the behavior of the model near the critical regime. On the other hand, for practical applications it may crucial to be able to predict the size of $\zetas(p)$ based on estimates of the degree distribution $p$.
This paper aims to obtain qualitative insight into this question by studying properties of the functional $p \mapsto \zetas(p)$ in detail by analyzing its sensitivity to the variability of $p$.
\cred{
As a robust tool for comparing levels of variability, we apply stochastic convex ordering techniques. Our main results are Theorem~\ref{the:OrderingSmallExtinction}, which confirms that a higher degree variability decreases the limiting component size when certain special cases are ruled out, and Theorem~\ref{the:main2} which confirms the same for sufficiently supercritical mixed Poisson distributions with heavy tails. We also provide counterexamples which show that, rather counterintuitively, a higher degree variability may lead to a larger connected component in some special cases. Despite the vast literature on the asymptotics of configuration models (cf.\ \cite{VanDerHofstad_2014_Vol1} and references therein) and numerous works on the stochastic ordering properties of branching processes (eg.\ \cite{Pellerey_2007,Sawaya_Klaere_2014,Valdes_Yera_Zuaznabar_2014}), this paper appears to be the first of its kind to study the size biasing effects prominent in most random graph models of interest using stochastic ordering techniques.

The rest of the paper is organized as follows. In Section~\ref{Branching functional of the configuration model} we introduce notations and recollect basic results of size-biased distributions relevant to configuration models. Section~\ref{BO} summarizes various stochastic ordering notions related to branching processes. The key contributions of the paper are in Section~\ref{the:OrderinCM}, which contains the main results and their proofs, together with counterexamples and numerical experiments. Section~\ref{sec:Conclusions} concludes the paper.
}

\section{Branching functional of the configuration model} 
\label{Branching functional of the configuration model}
\subsection{Size biasing and downshifting}

The configuration model, like many real-world networks, exhibits a size-bias phenomenon in degrees, in that ``your friends are likely to have more friends than you do''. The \emph{size biasing} of a probability distribution $\mu$ on the nonnegative real line $\R_+$ (or a subset thereof) with mean $m_1(\mu) = \int x \mu(dx) \in (0,\infty)$, is the probability distribution $\mu^*$ defined by
\[
 \mu^*(B) = \frac{\int_B x \mu(dx)}{m_1(\mu)}, \quad B \subset \R_+.
\]
\cred{
Furthermore, the \emph{downshifted size biasing} of $\mu$, denoted $\mu^\circ$, is defined as the distribution of $X^\circ = X^*-1$ where $X^*$ is a random number distributed according to $\mu^*$. Note that if $X$ and $X^*$ are random numbers with distributions $\mu$ and $\mu^*$, respectively, then
\begin{equation}
 \label{eq:SizeBiasing}
 \E \, \phi(X^*) = \frac{\E \, \phi(X) X}{\E X}
\end{equation}
for any real function $\phi$ such that the above expectations exist. 

For probability distributions on $\Z_+=\{0,1,2,\dots\}$ we use the same symbol $p$ both for the distribution and its probability mass function. Then the size biasing and the downshifted size biasing of $p$ can be represented as
\begin{equation}
 \label{eq:DownshiftedSizebiasing}
 p^*(k) = \frac{k p(k)}{m_1(p)} 
 \qquad \text{and} \qquad
 p^\circ(k) = p^*(k+1).
\end{equation}
}
If $G_p(s) = \sum_{k \ge 0} s^k p(k)$ denotes the generating function of $p$, then the generating functions of $p^*$ and $p^\circ$ are given by 
\begin{equation}
 \label{eq:GeneratingFunctions}
 G_{p^*}(s) = \frac{s}{m_1(p)} G_p^{'}(s)
 \qquad\text{and}\qquad
 G_{p^\circ}(s) = \frac{1}{m_1(p)} G_p^{'}(s).
\end{equation}

\cred{
Table~\ref{tab:SizeBiasing} below summarizes the size biasings of some common probability distributions that are used in the sequel. Here $\delta_x$ denotes the Dirac point mass at $x$, $\Bin(n,p)$ and $\Poi(c)$ refer to the standard binomial and Poisson distributions, and $\MPoi(\mu)$ denotes the $\mu$-mixed Poisson distribution on $\Z_+$ with probability mass function
\[
 p(k) = \int_{\R_+} e^{{-\lambda}} \frac{\lambda^k}{k!} \, \mu(d\lambda), \quad k \in \Z_+.
\]
Moreover, $\Par(\alpha,c)$ is the Pareto distribution with shape $\alpha > 1$ and scale $c > 0$ (density function $f(t) = \alpha c^\alpha t^{-\alpha-1}$, $t > c$), and $\LogN(b,\sigma^2)$ is the lognormal distribution with location $b \in \R$ and scale $\sigma > 0$ (density function  $f(t) = \frac{1}{t\sqrt{2\pi \sigma^2}}\exp(-\frac{(\log t - b)^2}{2\sigma^2})$ for $t>0$ ) \cite[Equation (45)]{Arratia_Goldstein_Kochman_2015}.
}

\begin{table}[h]
\begin{center}
\setfracpadding
\cred{
\begin{tabular}{lll}
\toprule
Distribution \qquad ~ & Size biasing \quad ~ & Downshifted size biasing \\
\midrule
$\delta_x$ & $\delta_x$ & $\delta_{x-1}$ \\
$\Bin(n,p)$ & $\Bin(n-1,p) + 1$ & $\Bin(n-1,p)$ \\
$\Poi(c)$ & $\Poi(c)+1$ & $\Poi(c)$ \\
$\MPoi(\mu)$ & $\MPoi(\mu^*)+1$ & $\MPoi(\mu^*)$ \\
$\Par(\alpha,c)$ & $\Par(\alpha-1,c)$ & $\Par(\alpha-1,c)-1$ \\
$\LogN(b,\sigma^2)$ & $\LogN(b+\sigma^2, \sigma^2)$ & $\LogN(b+\sigma^2, \sigma^2) - 1$\\
\bottomrule
\end{tabular}
}
\end{center}
\caption{\label{tab:SizeBiasing} Size biasings of some common probability distributions ($\mu \pm 1$ is shorthand for the distribution of $X \pm 1$ with $X$ distributed according to $\mu$).}
\end{table}

\cred{In Section~\ref{sec:Social} we analyze in detail a class of Pareto-mixed Poisson distributions $\MPoi(\Par(\alpha,c))$ parameterized by $\alpha > 1$ and $c>0$. This class serves as a convenient benchmark for degree distributions of random graphs because it inherits the heavy-tailed behavior of the Pareto distribution, and because its downshifted size biasings are given in a simple form by
\begin{equation}
 \label{eq:ParPoi}
 \MPoi(\Par(\alpha,c))^\circ
 \weq \MPoi(\Par(\alpha,c)^*)
 \weq \MPoi(\Par(\alpha-1,c)).
\end{equation}
}

\subsection{Definition of the branching functional}

Given a probability distribution $p$ on $\Z_+$, we denote by
\[
 \eta(p) = \inf\{s \ge 0: G_p(s) = s\}
\]
the smallest fixed point of the generating function $G_p(s) = \sum_{k \ge 0} s^k p(k)$ in the interval $[0,1]$. Classical branching process theory (e.g.\ \cite{Grimmett_Stirzaker_2001,VanDerHofstad_2014_Vol1}) tells that $\eta(p) \in [0,1]$ is well defined and equal to the extinction probability of a Galton--Watson process with offspring distribution $p$. We denote the corresponding survival probability by
\begin{equation}
 \label{eq:Survival}
 \zeta(p) = 1 - \eta(p).
\end{equation}

As a consequence of \cite[Thm~2.3]{Janson_Luczak_2009}, the limiting maximum component size of a configuration model with limiting degree distribution $p$ corresponds to the survival probability of a two-stage branching process where the root node has offspring distribution $p$ and all other nodes have offspring distribution $p^\circ$ defined by \eqref{eq:DownshiftedSizebiasing}. Therefore, the branching functional $p \mapsto \zetas(p)$ appearing in  \eqref{eq:MaxComponent} can be written as
\begin{equation}
 \label{eq:zetas}
 \zetas(p) = 1 - G_p(\eta(p^\circ)).
\end{equation}

The following two elementary examples will serve to illustrate the nonconvexity for the branching functional (see Remark~\ref{rem:Nonconvex}).

\begin{example}
For the Dirac measure $\delta_n$ at an integer $n \ge 1$ we have
\begin{equation}
 \label{eq:Degenerate}
 \zeta_{\rm CM}(\delta_n)
 = \begin{cases}
 0, &\quad n =1, \\
 1, &\quad n \ge 2.
 \end{cases}
\end{equation}
To verify this it suffices to note that $\zetas(\delta_n) = 1 - G_{\delta_n}(\eta(\delta_n^\circ)) = 1 - \eta(\delta_n^\circ)^n$ where $\eta(\delta_n^\circ) = \eta(\delta_{n-1})$ equals one for $n=1$ and zero for $n \ge 2$.
\end{example}

\begin{example}
For the uniform distribution on $\{1,n\}$ we have
\begin{equation}
 \label{eq:Binary}
 \zetas \left( \frac{1}{2} \delta_1 + \frac{1}{2} \delta_n \right)
 \weq \begin{cases}
 0, &\quad n = 1,2, \\
 1-o(1), &\quad n \gg 2.
 \end{cases}
\end{equation}
To verify this, an elementary computation shows that the downshifted size biasing of $p_n =  \frac{1}{2} \delta_1 + \frac{1}{2} \delta_n$ equals
\[
 p_n^\circ \weq \left( \frac{1}{n+1} \right) \delta_0 +  \left( \frac{n}{n+1} \right) \delta_{n-1}.
\]
For $n =1,2$ the support of $p_n^\circ$ is contained in $\{0,1\}$ which implies that $\eta(p_n^\circ) = 1$. On the other hand, because $G_{p_n^\circ}(s) \to 0$ uniformly on $[0,s_0]$ for all $s_0 < 1$, it follows that $\eta(p_n^\circ) = o(1)$ as $n \to \infty$. Hence the fact that $G_{p_n}(s) \le s$ implies $G_{p_n}(\eta(p_n^\circ)) = o(1)$, and we may conclude \eqref{eq:Binary}.
\end{example}

\begin{remark}
\label{rem:Nonconvex}
The functional $p \mapsto \zeta_{\rm CM}(p)$ is not convex nor concave. To see why, consider the probability distribution $\alpha p + (1-\alpha) q$ where $\alpha = 1/2$, $p = \delta_1$, and $q=\delta_n$ for some integer $n \ge 1$. Then by formula \eqref{eq:Degenerate},
\[
 \alpha \zetas(p) + (1-\alpha) \zetas(q)
 \weq \frac{1}{2}
 \quad \text{for all} \ n \ge 2,
\]
whereas by formula \eqref{eq:Binary},
\[
 \zetas( \alpha p + (1-\alpha)q)
 \weq
 \begin{cases}
 0, &\quad n = 1,2, \\
 1-o(1), &\quad n \gg 2.
 \end{cases}
\]
\end{remark}

\subsection{Continuity properties}

\cred{We will next prove a continuity property of the branching functional which is needed for proving Theorem~\ref{the:main2}, one of the main results of this paper.} For probability distributions $\mu,\mu_1,\mu_2,\dots$ on $\R_+$, we denote convergence in distribution by $\mu_n \dto \mu$. For probability distributions on $\R_+$ with finite first moments, we denote $\mu_n \Wto \mu$ if $\mu_n$ converges to $\mu$ in distribution and $(\mu_n)$ is uniformly integrable. Recall \cite[Theorem 7.1.5]{Ambrosio_Gigli_Savare_2008} that $\mu_n \Wto \mu$ is equivalent to convergence in the Wasserstein metric defined by
\[
 d_{W}(\mu,\nu) \weq \inf \int_{\R_+ \times \R_+} |x-y| \, \gamma(dx,dy),
\]
where the infimum is computed over the set of all couplings of $\mu$ and $\nu$. See \cite{Leskela_Vihola_2013} for a probabilistic discussion about the Wasserstein metric.

\begin{theorem}
\label{the:ContinuityCM}
Let $p,p_n$ be probability measures on $\Z_+$ each with a finite nonzero mean. If $p_n \Wto p$ and $p(1) > 0$, then $\zetas(p_n) \to \zetas(p)$.
\end{theorem}

The proof of the theorem is based on the following two lemmas. The first states that convergence in the Wasserstein metric implies convergence in distribution for associated size biasings.

\begin{lemma}
\label{the:BiasedConvergence}
Let $\mu,\mu_1,\mu_2,\dots$ be probability measures on $\R_+$ each with a finite nonzero mean. If $\mu_n \Wto \mu$, then $\mu_n^* \dto \mu^*$ and $\mu_n^\circ \dto \mu^\circ$.
\end{lemma}
\begin{proof}
Uniform integrability and $\mu_n \dto \mu$ imply \cite[Lemma 4.11]{Kallenberg_2002} that $m(\mu_n) \to m(\mu)$. If $\phi: \R_+ \to \R$ is continuous and has compact support, then $\psi(x) = x \phi(x)$ is continuous and bounded, and hence
\[
 \mu_n^*(\phi)
 = \frac{\mu_n(\psi)}{m(\mu_n)}
 \to \frac{\mu(\psi)}{m(\mu)}
 = \mu^*(\phi).
\]
We conclude that $\mu_n^* \to \mu^*$ vaguely. In addition, the uniform integrability of $(\mu_n)$ implies tightness of $(\mu_n^*)$, and we may conclude \cite[Lemma 5.20]{Kallenberg_2002} that $\mu_n^* \dto \mu^*$. The fact that $\mu_n^\circ \dto \mu^\circ$ follows from the continuous mapping theorem \cite[Theorem 4.27]{Kallenberg_2002}.
\end{proof}

The following result on the continuity of extinction probabilities is probably well known. Because we did not find it in the literature, the proof is included in Appendix~\ref{sec:Appendix} for reader's convenience.

\begin{lemma}
\label{the:ContinuityOfExtinctionProbabilities}
Let $p,p_n$ be probability measures on $\Z_+$. If $p_n \dto p$ and $p(0) > 0$, then $\eta(p_n) \to \eta(p)$.
\end{lemma}

\begin{proof}[Proof of Theorem~\ref{the:ContinuityCM}]
By Lemma~\ref{the:BiasedConvergence}, $p_n^\circ \dto p^\circ$. Moreover, $p^\circ(0) = \frac{p(1)}{m_1(p)} > 0$. Hence $\eta(p_n^\circ) \to \eta(p^\circ)$ by Lemma~\ref{the:ContinuityOfExtinctionProbabilities}. The assumption that $p_n \dto p$ also implies that $G_{p_n} \to G_p$ uniformly on $[0,1]$, as explained in the proof of Lemma~\ref{the:ContinuityOfExtinctionProbabilities}.  Hence
\[
 \zetas(p_n)
 \weq 1- G_{p_n}(\eta(p_n^\circ)) 
 \ \to \ 1- G_{p}(\eta(p^\circ)) 
 \weq \zetas(p).
\]
\end{proof}

\subsection{Upper bounds}
\label{the:Upper}
A simple closed-form expression for $\zetas(p)$ is not readily available due to the implicit definition of $\eta(p^\circ)$. To get a qualitative insight into the behavior of $\zetas(p)$ as  a functional of $p$, analytical bounds will be valuable. The following result presents a fundamental upper bound which only depends on the mean degree distribution. This result is implicitly contained in the proof of \cite[Theorem 2]{Britton_Trapman_2012}. Here we provide a short and transparent proof.

\begin{proposition}
\label{the:CrudeUpper1}
For any probability distribution $p$ with a finite nonzero mean $\lambda$,
\begin{align}
\label{eq:CrudeUpper1}
 \zeta_{\rm CM}(p) \wle \frac{\lambda}{2}.
\end{align}
\end{proposition}
\begin{proof}
Denote $z = \eta(p^\circ)$, so that by definition,
$
 G_{p^\circ}(z) = z.
$
Moreover, the convexity of $G_{p^\circ}$ implies that $G_{p^\circ}(s) \le s$ for all $s \in [z,1]$. Next, by applying \eqref{eq:GeneratingFunctions} we can write
\[
 G_p(s) = 1-\lambda\int_s^1G_{p^\circ}(s) ds.
\]
Hence,
\[
 \zetas(p)
 \weq 1 - G_{ p}(z)
 \weq  \lambda \int_z^1 G_{p^\circ}(s) \, ds,
\]
and we may conclude that
\[
 \zetas(p)
 \wle  \lambda \int_z^1 s \, ds
 \wle \lambda \int_0^1 s \, ds
 \weq \frac{\lambda}{2}.
\]
\end{proof}

The upper bound in Proposition~\ref{the:CrudeUpper1} tells nothing for graphs of mean degree two or higher. The following result provides a crude upper bound applicable also for $\lambda \le 2$. Similar bounds for standard branching processes have been derived in \cite{Sawaya_Klaere_2014,Valdes_Yera_Zuaznabar_2014}.

\begin{proposition}
\label{the:CrudeUpper2}
For any probability distribution $p$ on $\Z_+$ with a finite nonzero mean $\lambda$,
\begin{equation}
 \label{eq:CrudeUpper2}
 \zetas(p) \ \le \ 1 - p(0) - \frac{p(1)^2}{\lambda}.
\end{equation}
\end{proposition}
\begin{proof}
Let $p^\circ$ be the downshifted size biasing of $p$ defined by \eqref{eq:DownshiftedSizebiasing}. Because a branching process with offspring distribution $p^\circ$ goes extinct at the first step with probability $p^\circ(0)$, it follows that
\[
 \eta(p^\circ) \ge p^\circ(0) = \frac{p(1)}{\lambda}.
\]
Together with $G_p(s) \ge p(0) + p(1) s$, this shows that
\[
 G_p(\eta(p^\circ)) \ge p(0) + \frac{p(1)^2}{\lambda}.
\]
The above inequality substituted into \eqref{eq:zetas} implies \eqref{eq:CrudeUpper2}.
\end{proof}

The following result provides a more accurate upper bound of $\zetas(p)$ based on $\lambda,p(0),p(1),p(2)$. Similar techniques may be applied to derive more accurate upper bounds when a larger collection of low values of the the probability mass function of $p$ are known.

\begin{proposition}
\label{the:CrudeUpper3}
For any probability distribution $p$ on $\Z_+$ with a finite nonzero mean $\lambda$,
\begin{equation}
 \label{eq:CrudeUpper3}
 \zetas(p) \ \le \ 1 - p(0) - p(1) a - p(2) a^2
\end{equation}
where $a = \frac{p(1)}{\lambda - 2 p(2)}$.
\end{proposition}
\begin{proof}
Observe that $G_{p^\circ}(s) \ge p^\circ(0) + p^\circ(1) s$ implies
\[
 \eta(p^\circ)
 \weq G_{p^\circ}(\eta(p^\circ))
 \wge p^\circ(0) + p^\circ(1) \eta(p^\circ),
\]
so that 
\[
 \eta(p^\circ)
 \wge 
 \frac{p^\circ(0)}{1 - p^\circ(1)}
 \weq \frac{\lambda^{-1}p(1)}{1 - 2 \lambda^{-1} p(2)}
 \weq a.
\]
Then by \eqref{eq:zetas} and the monotonicity of $G_p$ we find that
\[
 \zetas(p)
 \weq 1 - G_p(\eta(p^\circ))
 \wle1 - G_p(a).
\]
Hence the claim follows by $G_p(a) \ge p(0) + p(1) a + p(2) a^2$.
\end{proof}

\subsection{Thinning}
\label{sec:Thinning}

The study of percolation and epidemics on random graphs requires the analysis of thinned degree distributions (see Section~\ref{sec:Social}). The \emph{$r$-thinning} of a probability distribution $p$ on $\Z_+$ with $r \in [0,1]$ is the probability distribution $T_r p$ on $\Z_+$ with probability mass function
\[
 T_r p(k)
 \weq \sum_{\ell \ge k} p(\ell) \binom{\ell}{k} r^{k} (1-r)^{\ell - k}
\]
and generating function
\[
 G_{T_r p}(s) \weq G_p \circ G_{\Ber(r)}(s) \weq G_p(1-r+rs).
\]
The $r$-thinning of $p$ can be recognized as a mixed binomial distribution of a random integer $X_r$ corresponding to a random vector $(X_r,X)$ where $X$ is $p$-distributed and the conditional distribution of $X_r$ given $X = n$ is $\Bin(n,r)$.  Alternatively, if $X,\theta_1,\theta_2,\dots$ are mutually independent and such that $X$ is $p$-distributed and $\theta_i$ is $\Ber(r)$-distributed, then $X_r = \sum_{i=1}^X \theta_i$ is distributed according to the $r$-thinning of $p$.

\begin{example}
The $r$-thinnings of Dirac, binomial, and Poisson distributions are given by $T_r(\delta_n) = \Bin(n,r)$, $T_r(\Bin(n,\alpha)) = \Bin(n,\alpha r)$, and $T_r(\Poi(\lambda)) = \Poi(\lambda r)$.
\end{example}

\begin{lemma}
\label{the:Commuting}
The downshifted size biasing and $r$-thinning operations commute according to $(T_r p)^\circ = T_r(p^\circ)$.
\end{lemma}
\begin{proof}
Because $G_{\Ber(r)}'(s) = \frac{d}{ds} (1-r+rs) = r$, we find that
\[
 G_{T_r p}'(s) = G_p'(G_{\Ber(r)}(s)) \, r.
\]
Using this formula together with \eqref{eq:GeneratingFunctions} and $m(T_r p) = r m_1(p)$, we see that
\[
 G_{(T_r p)^\circ}(s)
 = \frac{G_{T_r p}'(s)}{m(T_r p)}
 = \frac{G_p'(G_{\Ber(r)}(s)) \, r}{r m_1(p)}
 = \frac{G_p'(G_{\Ber(r)}(s))}{m_1(p)},
\]
and from this we may conclude that
\[
 G_{(T_r p)^\circ}(s)
 = G_{p^\circ}(G_{\Ber(r)}(s))
 = G_{T_r(p^\circ)}(s).
\]
\end{proof}

\section{Stochastic ordering of branching processes}
\label{BO}

\subsection{Strong and convex stochastic orders}

The upper bound of $\zetas(p)$ obtained in Proposition~\ref{the:CrudeUpper2} is rough as it disregards information about the tail characteristics of $p$. To obtain better estimates, we will develop in this section techniques based on the theory of stochastic orders (see \cite{Muller_Stoyan_2002} or \cite{Shaked_Shanthikumar_2007} for comprehensive surveys).

Integral stochastic orderings between probability distributions on $\R$ (or a subset) are defined by requiring
\begin{equation}
\label{eq:IntegralOrder}
 \int \phi(x) \mu(dx) \le \int \phi(x) \nu(dx)
\end{equation}
to hold for all functions $\phi: \R \to \R$ in a certain class of functions such that both integrals above exist. The \emph{strong stochastic order} is defined by denoting $\mu \lest \nu$ if \eqref{eq:IntegralOrder} holds for all increasing functions $\phi$. The \emph{convex stochastic order} (resp.\ concave, increasing convex, increasing concave) order is defined by denoting $\mu \lecx \nu$ (resp.\ $\mu \lecv \nu$, $\mu \leicx \nu$ $\mu \leicv \nu$) if \eqref{eq:IntegralOrder} holds for all convex (resp.\ concave, increasing convex, increasing concave) functions $\phi$. For random numbers $X$ and $Y$ distributed according to $\mu$ and $\nu$, we denote $X \lest Y$ if $\mu \lest \nu$, and similarly for other integral stochastic orders.

When $X \lest Y$ we say that $X$ is smaller than $Y$ in the strong order because then $\pr(X > t) \le \pr(Y > t)$ for all $t$. When $X \lecx Y$ we say that $X$ is less variable than $Y$ in the convex order, because then $\E X = \E Y$ and $\var(X) \le \var(Y)$ whenever the second moments exist. Note that $X \lecv Y$ if and only if $X \gecx Y$, that is, $X$ is less concentrated than $Y$. The order $X \leicv Y$ can be interpreted by saying that $X$ is smaller and less concentrated than $Y$.

\subsection{Stochastic ordering and branching processes}

To obtain sharp results for branching processes, it is useful to introduce one more integral stochastic order. 
For probability distributions $\mu$ and $\nu$ on $\R_+$ (or a subset thereof), the \emph{Laplace transform order} is defined by denoting $\mu \leLt \nu$ if \eqref{eq:IntegralOrder} holds for all functions $\phi$ of the form $\phi(x) = - e^{-tx}$ with $t \ge 0$. Observe that $\mu \leLt \nu$ is equivalent to requiring $L_\mu(t) \ge L_\nu(t)$ for all $t \ge 0$, where
we denote the Laplace transform of $\mu$ by
$
 L_\mu(t) = \int e^{-t x} \mu(dx).
$
For probability distributions $p$ and $q$ on $\Z_+$, observe that $p \leLt q$ if and only if their generating functions are ordered by $G_p(s) \ge G_q(s)$ for all $s \in [0,1]$. Because for any $t \ge 0$, the function $x \mapsto - e^{-tx}$ is increasing and concave, it follows that
\[
 \mu \lest \nu \implies \mu \leicv \nu \implies \mu \leLt \nu.
\]
Due to the above implications we may interpret $X \leLt Y$ as $X$ being smaller and less concentrated than $Y$ (in a weaker sense than $X \leicv Y$).

The following elementary result confirms an intuitive fact that a branching population with a smaller and more variable offspring distribution is less likely to survive in the long run. The proof can be obtained as a special case of a slightly more general result below (Lemma~\ref{the:ExtOrder}).

\begin{proposition} 
\label{OE}
When $p \leLt q$, the survival probabilities defined by \eqref{eq:Survival}
are ordered according to $\zeta(p) \le \zeta(q)$. Especially,
\[
 p \lest q \ \text{or} \ p \lecv q
 \ \implies \ p \leicv q
 \ \implies \ p \leLt q
 \ \implies \ \zeta(p) \le \zeta(q).
\]
\end{proposition}

\section{Stochastic ordering of the configuration model}
\label{the:OrderinCM}
Basic intuition about standard branching processes, as confirmed by Proposition~\ref{OE}, suggests that a large configuration model  with a smaller and more variable degree distribution should have a smaller giant component. The next subsection displays a counterexample where this intuitive reasoning fails.

\subsection{A counterexample}
\label{exa:Special}

Consider degree distributions $p$ and $q$ defined by
\begin{align*}
 p &= \frac{1}{8} \delta_1 + \frac{6}{8} \delta_2 + \frac{1}{8} \delta_3, \\
 q &= \frac{1}{16} \delta_0 + \frac{1}{8} \delta_1 + \frac{5}{8} \delta_2 + \frac{1}{8} \delta_3 + \frac{1}{16} \delta_4,
\end{align*}
where $\delta_k$ represents the Dirac point mass at point $k$. Their downshifted size biasings, computed using \eqref{eq:DownshiftedSizebiasing}, are given by
\begin{align*}
 p^\circ &= \frac{1}{16} \delta_0 + \frac{12}{16} \delta_1 + \frac{3}{16} \delta_2, \\
 q^\circ &= \frac{1}{16} \delta_0 + \frac{10}{16} \delta_1 + \frac{3}{16} \delta_2 + \frac{2}{16} \delta_3.
\end{align*}

By comparing integrals of cumulative distributions functions \cite[Thm 3.A.1]{Shaked_Shanthikumar_2007} or by constructing a martingale coupling \cite{Leskela_Vihola_2014-04}, it is not hard to verify that in this case $p \lecx q$. Numerically computed values for the associated means, variances, and extinction probabilities are listed in Table~\ref{tab:1}.
\begin{table}
\begin{center}
\begin{tabular}{ c || c | c || c | c }
 & $p$ & $q$ & $p^\circ$ & $q^\circ$ \\ \hline
mean & 2.000 & 2.000 & 1.125 & 1.375 \\ 
variance & 0.250 & 0.750 & 0.234 & 0.609 \\ 
extinction probability $\eta$ & 0.000 & 0.076 & 0.333 & 0.186 \\
\hline
\end{tabular}
\end{center}
\caption{Statistical indices associated to $p$ and $q$ and their downshifted size biasings. \label{tab:1}}
\end{table}
By evaluating the associated generating functions at $\eta(p^\circ) = 0.333$ and $\eta(q^\circ) = 0.186$, we find using \eqref{eq:zetas} that $\zetas(p) = 0.870$ and $\zetas(q) = 0.892$.

This example shows that a standard branching process with a less variable offspring distribution ($p \lecx q$) is less likely to go extinct ($\eta(p) < \eta(q)$), but the same is not true for the downshifted size-biased offspring distributions ($\eta(p^\circ) > \eta(q^\circ))$. As a consequence, the giant component of a large random graph corresponding to a configuration model with limiting degree distribution $p$ is with high probability smaller than the giant component in a similar model with limiting degree distribution $q$, even though $p$ is less variable than $q$. The reason for this is that, even though higher variability has an unfavorable effect on standard branching (the immediate neighborhood of the root note), higher variability also causes the neighbors of a neighbor to have bigger degrees on average.

\subsection{A monotonicity result when one extinction probability is small}
\label{the:monotonicity}
The following result shows that increasing the variability of a degree distribution $p$ \emph{does} decrease the limiting relative size of a giant component, under the extra conditions that $p(0) = q(0)$ and that the extinction probability related to $q^\circ$ is less than $e^{-2} \approx 0.135$. Note that in the analysis of configuration models it is often natural to assume that $p(0) = q(0)$ because nodes without any half-links have no effect on large components.

\begin{theorem}
\label{the:OrderingSmallExtinction}
Assume that $p \leicv q$, $p(0) = q(0)$, and $\eta(q^\circ) \le e^{-2}$. Then $\zetas(p) \le \zetas(q)$.
\end{theorem}

\begin{remark}
Assume that $q(1) > 0$ and that $\zetas(q) \ge 1 - q(0) - q(1) e^{-2}$. If this holds, then the inequality $G_q(s) \ge q(0) + q(1)s$ applied to $s = \eta(q^\circ)$ implies that
\[
 q(0) + q(1) e^{-2} \ge 1 - \zetas(q) = G_q(\eta(q^\circ)) \ge q(0) + q(1) \eta(q^\circ),
\]
so that $\eta(q^\circ) \le e^{-2}$ as required in Theorem~\ref{the:OrderingSmallExtinction}. 
\end{remark}

Theorem~\ref{the:OrderingSmallExtinction} is a direct consequence (choose $\ell=0$ below) of the following slightly more general result.

\begin{theorem}
\label{the:GenereralOrderingSmallExtinction}
Assume that $p \leicv q$, $p(i) = q(i)$ for all $i\in\{0,1,2,...,\ell\}$ and $\eta(q^\circ) \le e^{-2/(\ell+1)}$ for some integer $\ell\ge 0$. Then $\zetas(p) \le \zetas(q)$.
\end{theorem}

The proof of Theorem~\ref{the:GenereralOrderingSmallExtinction} is based on the following two lemmas.

\begin{lemma}
\label{the:sbgf}
If $p \leicv q$ and  $p(i) = q(i)$ for all $i\in\{0,1,2,...,\ell\}$, then the generating functions of the downshifted size biasings of $p$ and $q$ are ordered by
\[
 G_{p^\circ}(s) \ge G_{q^\circ}(s) \quad \text{for all $s \in [0,e^{-2/(\ell+1)}]$}.
\]
\end{lemma}

\begin{proof}
Fix $s\in(0,e^{-2/(\ell+1)}]$. Define a function $\phi_s: \R_+ \to \R_+$ by
\[\phi_s(x) = x s^x.\]
Denote $t = -\log s$, so that $t \in [2/(\ell+1), \infty)$. Because $\phi_s'(x) = (1-tx) e^{-tx}$ and $\phi_s''(x) = (t x-2)t e^{-t x}$, we find that $\phi_s$ is decreasing on $[\frac{1}{-\log s},\infty)$ and convex on $[\frac{2}{-\log s},\infty)$. Because $-\log s\ge \frac{2}{\ell +1}$, it follows that $\phi_s$ is decreasing and convex on $[\ell +1,\infty).$ Define a decreasing convex function by
\begin{align}
\label{the:DecConvex}
\psi_s(x) = f_s(x) 1_{[0,x_0)}+\phi_s(x) 1_{[x_0,\infty)}(x)
\end{align}
where $f_s(x)=\phi_s(x_0)+\phi_s'(x_0)(x-x_0)$ and $x_0=\ell+1$ (see Figure \ref{fig:ConvexExtension}).

Let $X$ and $Y$ be random integers distributed according to $p$ and $q$. By assumption $p\leicv q$,
\[
 \E\big(\psi_s(X)\big)\ge \E\big(\psi_s(Y)\big).
\]
By the second assumption $p(i) = q(i)$ for all  $i\in\{0,1,2,...,\ell\}$, the above inequality can also be written as
\[
 \sum_{i=\ell +1}^\infty\psi_s(i)p(i)\ge\sum_{i=\ell+1}^\infty\psi_s(i)q(i),
\]
and hence by \eqref{the:DecConvex} we find that
\[
 \sum_{i=\ell +1}^\infty\phi_s(i)p(i)\ge \sum_{i=\ell+1}^\infty\phi_s(i)q(i).
\]
By applying again the second assumption we obtain
\[
 \E\big(\phi_s(X)\big)\ge\E\big(\phi_s(Y)\big),
\]
which implies
$
 G_p'(s)\ge G_q'(s).
$
Because $p\leicv q$ also implies $\E(X)\le \E(Y)$, we see by \eqref{eq:GeneratingFunctions} that
\begin{align}
 \label{ineq}
 G_{p^{\circ}}(s)=\frac{G_p'(s)}{\E(X)}\ge \frac{G_q'(s)}{\E(Y)}=G_{q^{\circ}}(s).
\end{align} 

Hence the claim is true for all $s\in(0,e^{-2/(\ell+1)}]$. By the continuity of $G_{p^\circ}$ and $G_{q^\circ}$, the claim is also true for $s=0$.
\end{proof}
  
\begin{figure}[h]
\begin{center}
\includegraphics[height=.4\textwidth]{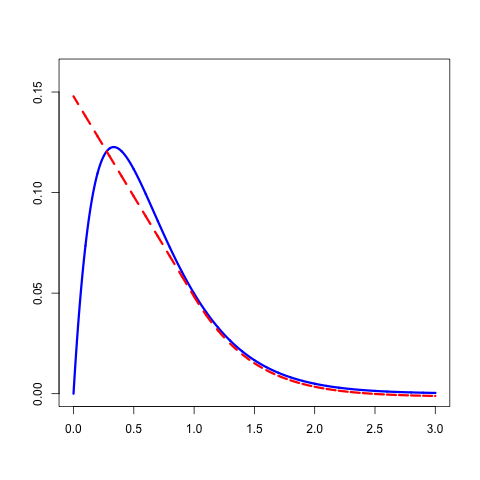}
\end{center}
\caption{\label{fig:ConvexExtension} Function $\phi_s$ (blue) and its convex modification $\psi_s$ (red) for $t=3$.}
\end{figure}

\begin{lemma}
\label{the:ExtOrder}
If $G_p(s) \ge G_q(s)$ for all $s \in [0,\eta(q)]$, then $\eta(p) \ge \eta(q)$.
\end{lemma}
\begin{proof}
The claim is trivial for $\eta(q) = 0$, so let us assume that $\eta(q) > 0$. Then $G_q(0) > 0$, and the continuity of $s \mapsto G_q(s)-s$ implies that $G_q(s) > s$ for all $s \in [0,\eta(q))$. Hence also
\[
 G_p(s) \ge G_q(s) > s
\]
for all $s \in [0,\eta(q))$. This shows that $G_p$ has no fixed points in $[0,\eta(q))$ and therefore $\eta(p)$, the smallest fixed point of $G_p$ in $[0,1]$, must be greater than or equal to $\eta(q)$.
\end{proof}

\begin{proof}[Proof of Theorem \ref{the:GenereralOrderingSmallExtinction}]
By applying Lemma \ref{the:sbgf} we see that
\begin{equation}
 \label{eq:sbOrder}
 G_{p^\circ}(s) \ge G_{q^\circ}(s)
\end{equation}
for all $s \in [0,e^{-2/(\ell+1)}]$. The assumption $\eta(q^\circ) \le e^{-2/(\ell+1)}$ further guarantees that \eqref{eq:sbOrder} is true for all $s \in [0,\eta(q^\circ)]$. Lemma~\ref{the:ExtOrder} then shows that $\eta(p^\circ) \ge \eta(q^\circ)$. Finally, $p \leicv q$ implies $p \leLt q$, so that $G_p(s) \ge G_q(s)$ for all $s \in [0,1]$. Therefore, the monotonicity of $G_p$ implies that
\[
 G_p(\eta(p^\circ)) \ge G_p(\eta(q^\circ)) \ge G_q(\eta(q^\circ)).
\]
By substituting the above inequality into \eqref{eq:zetas}, we obtain Theorem~\ref{the:GenereralOrderingSmallExtinction}.
\end{proof}

\subsection{Application: Social network modeling}
\label{sec:Social}

Consider a large online social network of mean degree $\lambda_0$ where users forward copies of messages to their neighbors independently of each other with probability $r_0$. Without any a priori information about the higher order statistics of the degree distribution, one might choose to model the network using a configuration model with some degree distribution which is similar to one observed in some known social network. Because several well-studied social networks data exhibit a power-law tail in their degree data, a natural first choice is to model the unknown network using a configuration model with a Pareto-mixed Poisson limiting degree distribution (recall~\eqref{eq:ParPoi} and Table~\ref{tab:SizeBiasing})
\begin{equation}
 \label{eq:NetworkStructure}
 p_0 = \MPoi( \Par(\alpha, c_0 ))
\end{equation}
with shape $\alpha > 1$, scale $c_0 = \lambda_0(1-1/\alpha)$, and mean $\lambda_0$.

Because the above choice of degree distribution was made without regard to network data, it is important to try to analyze how big impact can a wrong choice make to key network characteristics. When interested in global effects on information spreading, it is natural to consider the epidemic generated graph obtained by deleting stubs of the initial configuration model independently with probability $1-r_0$. The outcome corresponds to another configuration model where the limiting degree $p = T_{r_0} p_0$ is the $r_0$-thinning of $p_0$  defined in Section~\ref{sec:Thinning}. Using generating functions one may verify that the $r$-thinning of a $\mu$-mixed Poisson distribution $\MPoi(\mu)$ equals $\MPoi(r\mu)$, where $r \mu$ denotes the distribution of a $\mu$-distributed random number multiplied by $r \in [0,1]$. Because $r \Par(\alpha,c) = \Par(\alpha,r c)$, it follows that the Pareto-mixed Poisson distribution is scale-free in the sense that
\[
 T_r \MPoi(\Par(\alpha,c)) = \MPoi(\Par(\alpha,r c)).
\]
See \cite{Arratia_Liggett_Williamson_2014} for an insightful discussion on scale-free properties of discrete probability distributions. As a consequence, the $r_0$-thinning of $p_0$ in \eqref{eq:NetworkStructure} equals
\begin{equation}
 \label{eq:Case}
 p = \MPoi( \Par(\alpha,\lambda(1-1/\alpha)) )
\end{equation}
with $\lambda = \lambda_0 r_0$.

Now the key quantity describing the information spreading dynamics of the social network model is given by $\zetas(p)$ defined in \eqref{eq:zetas}. To study how sensitive this functional is to the variability of $p$, we have numerically evaluated $\zetas(p)$ for different values of $\alpha$ and $\lambda$, see Fig.~\ref{fig:ParetoMixedPoisson}. An extreme case is obtained by letting $\alpha \to \infty$ which leads to the standard Poisson distribution with mean $\lambda$. Again, perhaps a bit surprisingly, we see that for small values of $\lambda$, a Pareto-mixed Poisson as a limiting degree distribution may produce an asymptotically larger maximally connected component in a configuration model than a one with a less variable unmixed Poisson distribution with the same mean. This phenomenon is most clearly visible when $\lambda=0.9$, in which case $\zetas(\Poi(\lambda)) = \zeta(\Poi(\lambda)) = 0$ but a Pareto-mixed Poisson degree distribution with a heavy enough tail yields nonzero values of $\zetas$, as shown by the magenta curve in Fig.~\ref{fig:ParetoMixedPoisson}. For sufficiently large values of $\lambda$, this phenomenon appears not to take place.

\begin{figure}[h]
\begin{center}
\includegraphics[height=.65\textwidth]{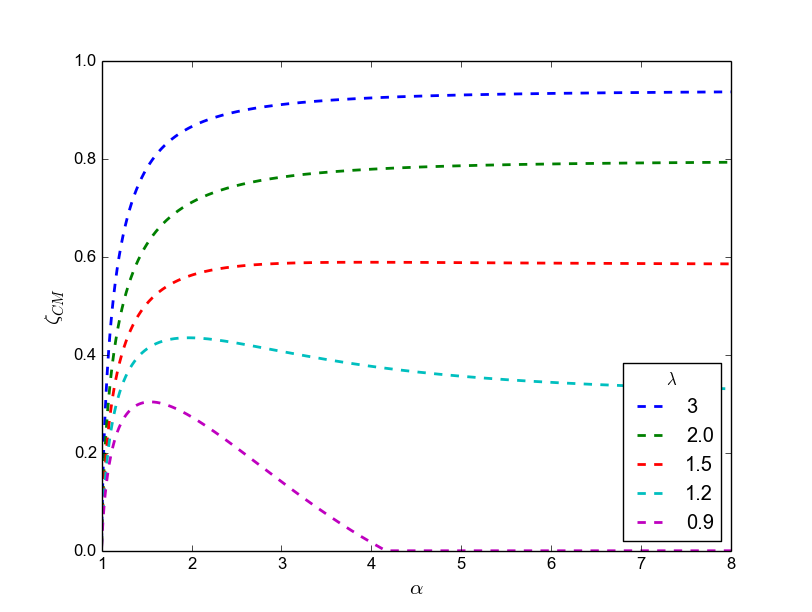}
\end{center}
\caption{\label{fig:ParetoMixedPoisson} Configuration model branching functional $\zetas(p)$ for a collection of Pareto-mixed Poisson degree distributions with mean $\lambda$, plotted as a function of the shape parameter $\alpha>1$.}
\end{figure}

Proving the monotonicity of $\zetas(p)$ for Pareto-mixed Poisson distributions of the form \eqref{eq:Case} is not directly possible using Theorem~\ref{the:OrderingSmallExtinction} because $p(0)$ is not constant with respect to the shape parameter $\alpha$. However, the following result can be applied here. Let us define a constant
\[
 \lambdacr = \inf\{\lambda \ge 0: \lambda \zeta(\Poi(\lambda)) = 2\}.
\]
Because $\lambda \mapsto \lambda \zeta(\Poi(\lambda))$ is increasing (Proposition~\ref{OE}) and continuous (Lemma~\ref{the:ContinuityOfExtinctionProbabilities}) and grows from zero to infinity as $\lambda$ ranges from zero to infinity, it follows that $\lambdacr \in (2,\infty)$ is well defined. Numerical computations indicate that $\lambdacr \approx 2.3$. The following result establishes a monotonicity result for the configuration model with a Pareto-mixed Poisson limiting distribution $p_\alpha = \MPoi(\mu_\alpha)$ with $\mu_\alpha = \Par(\alpha,c_\alpha)$, where the scale $c_\alpha = \lambda(1-1/\alpha)$ is chosen so that the mean of $p_\alpha$ equals $\lambda$ for all $\alpha > 0$ (recall Table~\ref{tab:SizeBiasing}).

\begin{theorem}
\label{the:main2}
For any $\lambda > \lambdacr$ there exists a constant $\alphacr > 1$ such that
\[
 \zetas(p_\alpha) \le \zetas(p_\beta) \le \zetas(\Poi(\lambda))
\]
for all $\alphacr \le \alpha \le \beta$.
\end{theorem}
\begin{remark}
Note that $\zetas(\Poi(\lambda)) = \zeta(\Poi(\lambda))$ due to the fact that the Poisson distribution is invariant to downshifted size biasing (cf.\ Table~\ref{tab:SizeBiasing}).
\end{remark}

\begin{proof}
Fix $\lambda > \lambdacr$ and denote $\eta_\infty = \eta(\Poi(\lambda))$. Because $\lambda > \lambdacr$, it follows that $\lambda(1-\eta_\infty) > 2$, and therefore
\begin{equation}
 \label{eq:BigAlpha1}
 \lambda(1-\eta_\infty)  \ge  \frac{2}{1-1/\alpha_0} + \lambda \epsilon
\end{equation}
for some large enough $\alpha_0 > 1$ and small enough $\epsilon > 0$. Next, Lemma~\ref{the:ParetoBigAlpha} below shows that $\mu_\alpha^* = \Par(\alpha-1,c_\alpha) \to \delta_\lambda$ and hence also $p_\alpha^\circ = \MPoi(\mu_\alpha^*) \to \Poi(\lambda)$ in distribution as $\alpha \to \infty$.
The continuity of the standard branching functional (Lemma~\ref{the:ContinuityOfExtinctionProbabilities}) implies that $\eta(p_\alpha^\circ) \to \eta_\infty$, and we may choose a constant $\alphacr \ge \alpha_0$ such that $\eta(p_\alpha^\circ) \le \eta_\infty + \epsilon$ for all $\alpha \ge \alphacr$.

Assume now that $\alphacr \le \alpha \le \beta$. Then by Lemma~\ref{the:ParetoOrdering} we know that
\begin{equation}
 \label{eq:ParetoOrder}
 \mu_\alpha \lecv \mu_\beta \lecv \delta_\lambda.
\end{equation}
Furthermore, $c_{\alpha_0} \le c_\alpha \le c_\beta$ implies that the supports of $\mu_\alpha$, $\mu_\beta$, and $\delta_\lambda$ are contained in $[c_{\alpha_0},\infty)$. Lemma~\ref{the:MixedPoissonIcv} below implies that  $G_{p_\alpha^\circ}(s) \ge G_{p_\beta^\circ}(s) \ge G_{\Poi(\lambda)}$ for all $s \in [0,s_0]$ where $s_0 = 1-2/c_{\alpha_0}$. Now \eqref{eq:BigAlpha1} shows that
\[
 s_0 = 1 - \lambda^{-1} \left( \frac{2}{1-1/\alpha_0} \right)
 \ge 1 - \lambda^{-1} \left( \lambda(1-\eta_\infty) - \lambda \epsilon \right)
 = \eta_\infty + \epsilon,
\]
and hence the interval $[0,s_0]$ contains both $[0,\eta_\infty]$ and $[0,\eta(p_\beta^\circ)]$. By applying Lemma~\ref{the:ExtOrder} twice, it follows that $\eta(p_\alpha^\circ) \ge \eta(p_\beta^\circ) \ge \eta(\Poi(\lambda)) = \eta_\infty$. 

On the other hand, inequality \eqref{eq:ParetoOrder} together with \cite[Thm~8.A.14]{Shaked_Shanthikumar_2007} implies that $\MPoi(\mu_\alpha) \leicv \MPoi(\mu_\beta) \leicv \Poi(\lambda)$. Especially, $p_\alpha \leLt p_\beta \leLt \Poi(\lambda)$, so that $G_{p_\alpha} \ge G_{p_\beta} \ge G_{\Poi(\lambda)}$ pointwise on $[0,1]$. This together with the monotonicity of the generating functions shows that
\[
 G_{p_\alpha}(\eta(p_\alpha^\circ)) \ge G_{p_\beta}(\eta(p_\beta^\circ)) \ge G_{\Poi(\lambda)}(\eta(\Poi(\lambda))),
\]
and the claim follows by substituting the above inequalities into \eqref{eq:zetas}. 
\end{proof}

\begin{lemma}
\label{the:MixedPoissonIcv}
Let $p = \MPoi(\mu)$ and $q = \MPoi(\nu)$ where $\mu \leicv \nu$. Assume that the supports of $\mu$ and $\nu$ are contained in an interval $[c,\infty)$ for some $c \ge 2$. Then $G_{p^\circ}(s) \ge G_{q^\circ}(s)$ for all $s \in [0,1-2/c]$.
\end{lemma}
\begin{proof}
Note first that for $G_{\MPoi(\mu)}(s) = L_{\mu}(1-s)$ and recall from Table~\ref{tab:SizeBiasing} that $\MPoi(\mu)^\circ = \MPoi(\mu^*)$. Hence $G_{p^\circ}(s) = L_{\mu^*}(1-s)$. Fix $s \in [0,1-2/c]$ and note that $G_{p^\circ}(s) = m_1(\mu)^{-1} \int \phi_s(x) \, \mu(dx)$, where $\phi_s(x) = x e^{-(1-s)x}$. Because $\phi_s'(x) = (1-(1-s)x) e^{-(1-s)x}$ and $\phi''_s(x) = (1-s)( (1-s)x-2 ) e^{-(1-s)x}$, it follows that the function $\phi_s$ is decreasing on $[\frac{1}{1-s},\infty)$ and convex on $[\frac{2}{1-s},\infty)$. Because $s \in [0,1-2/c]$, it follows that $\phi_s$ is decreasing and convex on the supports of $\mu$ and $\nu$. Therefore $\mu \leicv \nu$ implies $\int \phi_s d\mu \ge \int \phi_s d\nu$. Because $\mu \leicv \nu$ also implies that the first moments are ordered according to $m_1(\mu) \le m_1(\nu)$, we conclude that
\[
 G_{p^\circ}(s)
 \ =  \ m_1(\mu)^{-1} \int \phi_s \, d\mu
 \ \ge \ m_1(\nu)^{-1} \int \phi_s \, d\nu
 \ = \ G_{q^\circ}(s).
\]
\end{proof}

\begin{lemma}
\label{the:ParetoBigAlpha}
If $c_\alpha \to \lambda \ge 0$ as $\alpha \to \infty$, then $\Par(\alpha,c_\alpha) \to \delta_\lambda$.
\end{lemma}
\begin{proof}
Let $U$ be a uniformly distributed random number in $(0,1)$. Then $X_\alpha = c_\alpha (1-U)^{-1/\alpha}$ has $\Par(\alpha,c_\alpha)$ distribution for all $\alpha$. Because $c_\alpha \to \lambda$ and $ (1-U)^{-1/\alpha} \to 1$, it follows that $X_\alpha \to \lambda$ almost surely, and hence also in distribution.
\end{proof}

\subsection{Numerical experiments}
\label{E2}

After a detailed analysis of the configuration model branching functional for Pareto-mixed Poisson degree distributions, a natural question to ask is whether or not similar observations remain valid or other types of distributions as well. We studied this question by performing  numerical experiments on two classes of distributions: lognormally mixed Poisson distributions and binomial distributions.

In Figure~\ref{fig:LogNMixedPoisson} we have plotted numerically evaluated values of $\zetas(p)$ where $p = \MPoi(\LogN(b,\sigma^2))$ is a lognormally mixed Poisson distribution with scale $\sigma^2 > 0$ and location $b = \log \lambda - \sigma^2/2$ chosen so that the mean of $p$ equals $\lambda > 0$  (Table~\ref{tab:SizeBiasing}). The curves are plotted as functions of $1/\sigma^2$, so that variability decreases along the horizontal axis. The behavior of the branching functional is the qualitatively the same as for the Pareto-mixed case: for network models with a small mean degree, higher variability may dramatically increase the size of the largest component.

\begin{figure}[h]
\begin{center}
  \includegraphics[height=.65\textwidth]{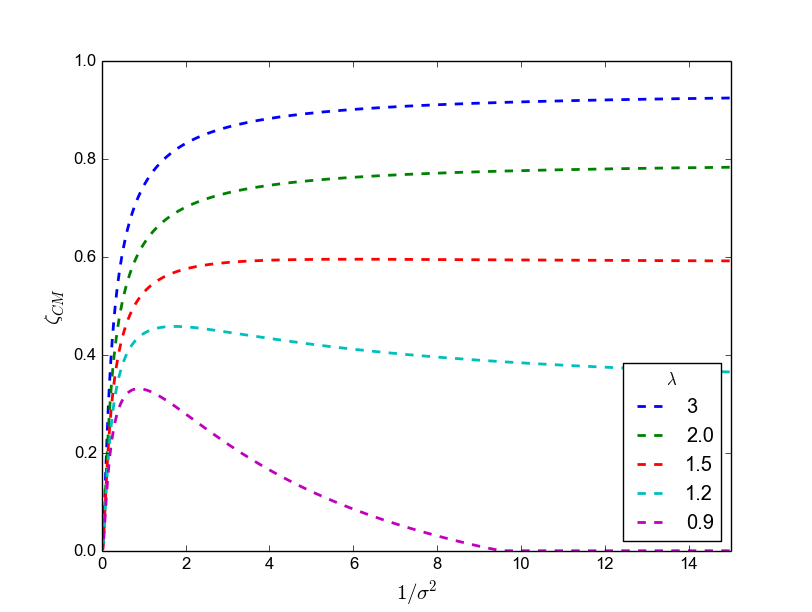}
\end{center}
\caption{\label{fig:LogNMixedPoisson} Configuration model branching functional $\zetas(p)$ for a collection of lognormally mixed Poisson distributions with mean $\lambda$, plotted as a function of $1/\sigma^2>0$.}
\end{figure}

In Figure~\ref{fig:Binomial} we have plotted numerically evaluated values of $\zetas(p)$ where $p = \Bin(n,\lambda/n)$ is a binomial distribution with mean $\lambda$, parameterized by $n \ge 3$. The variance of $p$ equals $\lambda(1-\lambda/n)$ and increases towards $\lambda$ along the horizontal axis. Also in this case with a light-tailed degree distribution, the overall qualitative picture is the same as for the Pareto and lognormally mixed Poisson distributions. The one difference is that all curves in Figure~\ref{fig:Binomial} appear to be monotone, either increasing (for small mean degree) or decreasing (for large mean degree). In addition, in this case the values of $\lambda \le 1$ always produce $\zetas(p) = 0$, because the downshifted size biasing of $\Bin(n,\lambda/n)$ equals $\Bin(n-1,\lambda/n)$ and has mean $\lambda(1-1/n) \le 1$ whenever $\lambda \le 1$.

\begin{figure}[h]
\begin{center}
\includegraphics[height=.65\textwidth]{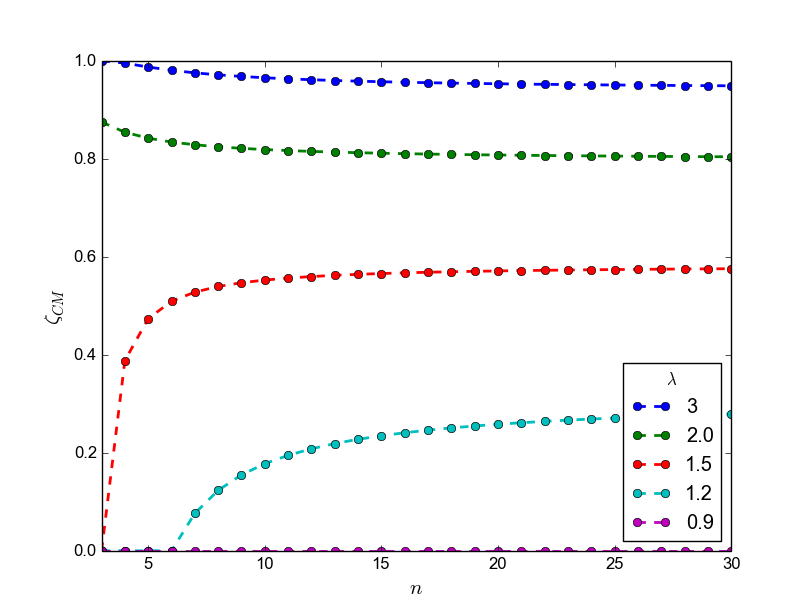}
\end{center}
\caption{\label{fig:Binomial} Configuration model branching functional $p \mapsto \zetas(p)$ for a collection of binomial degree distributions with mean $\lambda$, plotted as a function of $n\ge 3$.}
\end{figure}

\section{Conclusions}
\label{sec:Conclusions}
In this paper we studied the effect of degree variability to the global connectivity properties of large network models. The analysis was restricted to the configuration model and the associated uniform random graph with a given limiting degree distribution. Counterexamples were discovered both for a bounded support and power-law case that described that due to size biasing effects, increased degree variability may sometimes have a favorable effect on the size of the giant component, in sharp contrast to standard branching processes. We also proved using rigorous mathematical arguments that for certain natural classes of sufficiently supercritical network models, the increased degree variability has a negative effect on the global connectivity. Numerical experiments illustrate that these observations can be detected for both light-tailed and heavy-tailed degree distributions. Because most real-world social networks have mean degree much higher than one, we do not expect to encounter anomalous variability effects in their global connectivity structure. However, such effects might be important to take into account when studying long-range effects on epidemic generated graphs spanned by links over which a rare message or pathogen is transmitted.

\subsection*{Acknowledgements}

This research has been financially supported by the Emil Aaltonen Foundation. \cred{We thank two anonymous referees for helpful comments on improving the presentation of the paper.}

\appendix

\section{A continuity property of branching processes}
\label{sec:Appendix}
\begin{proof}[Proof of Lemma~\ref{the:ContinuityOfExtinctionProbabilities}]
We denote $\eta = \eta(p)$ and $\eta_n = \eta(p_n)$. We also denote $G(t) = G_p(t)$ and $G_n(t) = G_{p_n}(t)$. Observe first that
\[
 |G_n(t) - G(t)|
 \wle \sum_{k=0}^\infty t^k |p_n(k) - p(k)|
 \wle 2 \dtv(p_n,p).
\]
for all $t \in [0,1]$, where $\dtv$ refers to the total variation distance. Because convergence in total variation and convergence in distribution are equivalent on countable spaces, it follows that
\begin{equation}
 \label{eq:UniformConvergence}
  ||G_n - G||
  = \sup_{t \in [0,1]} |G_n(t) - G(t)|
  \to 0.
\end{equation}

(i) Consider first the case where $\eta(p) \in (0,1)$. Then $p(0) + p(1) < 1$. Hence $p(k) > 0$ for some $k \ge 2$, and this shows that
\[
 G''(t) = \sum_{k=2}^\infty p(k) k(k-1)t^{k-2} > 0
\]
for all $t \in (0,1)$. Note that by the continuity of $G(t)-t$, we see that
\begin{equation}
 \label{eq:ContinuitySurvival0}
 G(t) > t \quad \text{for all $t \in [0,\eta)$}.
\end{equation}

We will next show that
\begin{equation}
 \label{eq:ContinuitySurvival1}
 G(t) < t \quad \text{for all $t \in (\eta,1)$}.
\end{equation}
Assume, on the contrary, that $G(\eta') \ge \eta'$ for some $\eta' \in (\eta,1)$. Then indeed $G(\eta') = \eta'$, because by the convexity of $G$, the point $(\eta',G(\eta'))$ must lie below the line connecting the points $(\eta,G(\eta)) = (\eta,\eta)$ and $(1,G(1)) = (1,1)$. But then the points $(\eta,G(\eta), (\eta',G(\eta')), (1,G(1))$ all lie on the straight line between $(\eta,\eta)$ and $(1,1)$, and the convexity $G$ implies that $G(t) = t$ on $[\eta,1]$. This contradicts the fact that $G''(t) > 0$ on $(\eta,1)$. Hence we conclude that \eqref{eq:ContinuitySurvival1} is valid.

Next fix any $\epsilon > 0$ such that $(\eta-\epsilon, \eta+\epsilon) \subset (0,1)$. Then \eqref{eq:ContinuitySurvival0} and the continuity of $G(t) - t$ imply that $\delta_1 := \inf_{t \in [0,\eta-\epsilon]} (G(t)-t) > 0$. On the other hand \eqref{eq:ContinuitySurvival1} implies that $\delta_2 := (\eta+\epsilon) - G(\eta+\epsilon) > 0$. By \eqref{eq:UniformConvergence} we may fix $n_0$ such that $||G_n - G|| \le \min(\delta_1,\delta_2)/2$ for all $n \ge n_0$. Then for any $n \ge n_0$, we find that $G_n(t) - t \ge \delta_1/2$ on $[0,\eta-\epsilon]$ and $G_n(\eta+\epsilon) \le G(\eta+\epsilon) + ||G_n-G|| \le (\eta+\epsilon) - \delta_2/2$.
Hence $G_n(t) > t$ on $[0,\eta-\epsilon]$ and $G_n(t) < t$ for $t = \eta+\epsilon$. By the continuity of $G_n$, we conclude that $G_n$ has no fixed points in $[0,\eta-\epsilon]$ and at least one fixed point in $(\eta-\epsilon,\eta+\epsilon)$. Therefore the smallest nonzero fixed point of $G_n$ satisfies $\eta_n \in (\eta-\epsilon, \eta+\epsilon)$ for all $n \ge n_0$.

(ii) Consider next the case with $\eta = 1$. Then $G(0) > 0$ and the continuity of $G$ imply that $G(t) > t$ on the interval $(0,1)$. Especially, $\delta = \inf_{t \in [0,1-\epsilon]} G(t) > 0$ for any $\epsilon>0$. For any large enough $n$ such that $||G_n-G|| \le \delta/2$, it follows that $G_n(t) \ge G(t) - ||G_n-G|| \ge t + \delta/2$ for all $t \in [0,1-\epsilon]$. This shows that $G_n$ has no fixed points in $[0,1-\epsilon]$, and hence $\eta_n > 1-\epsilon$ for all big enough $n$.
\end{proof}

\section{Stochastic ordering of Pareto distributions}

The following result characterizes stochastic ordering properties of Pareto distributions. For $i=1,2$, let $\mu_i = \Par(\alpha_i,c_i)$ be the Pareto distribution with shape $\alpha_i > 1$, scale $c_i > 0$, and mean $\lambda_i = c_i(1-1/\alpha_i)^{-1}$.

\begin{theorem}
\label{the:ParetoOrdering}
For any Pareto distributions $\mu_i = \Par(\alpha_i,c_i)$ with shape $\alpha_i > 1$, scale $c_i > 0$, and mean $\lambda_i = c_i(1-1/\alpha_i)^{-1}$:
\begin{enumerate}[(i)]
\item $\mu_1 \leicx \mu_2$ if and only if $\lambda_1 \le \lambda_2$ and $\alpha_1 \ge \alpha_2$.
\item $\mu_1 \lecx \mu_2$ if and only if $\lambda_1 = \lambda_2$ and $\alpha_1 \ge \alpha_2$.
\item $\mu_1 \leicv \mu_2$ if and only if $\lambda_1 \le \lambda_2$ and $c_1 \le c_2$.
\end{enumerate}
\end{theorem}
Result (i) above quantifies the intuitively natural property that a larger mean and a heavier tail makes a Pareto distribution bigger and more variable in the increasing convex order. Interestingly, result (iii) may be valid for both $\alpha_1 < \alpha_2$ and $\alpha_1 > \alpha_2$, depending on the value of the scale parameter.

\begin{proof}
(i) Let $F_i^{-1}(s) = c_i (1-s)^{-1/\alpha_i}$ be the quantile function of $\mu_i$, and denote the upper and lower integrated quantile functions of $\mu_i$ by
\begin{align*}
 \bar H_i(t)
 &= \int_t^1 F_i^{-1}(s) \, ds = \lambda_i (1-t)^{1-1/\alpha_i}, \\
 H_i(t) &= \int_0^t F_i^{-1}(s) \, ds = \lambda_i (1- (1-t)^{1-1/\alpha_i}).
\end{align*}

Now by \cite[Thm~3.A.5, Thm 4.A.3]{Shaked_Shanthikumar_2007} $\mu_1 \leicx \mu_2$ if and only if $\bar H_1(t) \le \bar H_2(t)$ for all $t \in (0,1)$,  that is,
\begin{equation}
 \label{eq:ParetoUpper}
 \lambda_1 (1-t)^{1-1/\alpha_1}
 \le \lambda_2 (1-t)^{1-1/\alpha_2}
\end{equation}
for all $t \in (0,1)$. If $\lambda_1 \le \lambda_2$ and $\alpha_1 \ge \alpha_2$, the the above inequality holds for all $t \in (0,1)$, and hence $\mu_1 \leicx \mu_2$.

Assume next that $\mu_1 \leicx \mu_2$. Then  \eqref{eq:ParetoUpper} is valid for all $t \in (0,1)$. By letting $t \to 0$, it follows that $\lambda_1 \le \lambda_2$. Inequality \eqref{eq:ParetoUpper} also implies that the fraction
\[
 \frac{(1-t)^{1-1/\alpha_1}}{(1-t)^{1-1/\alpha_2}}
 = (1-t)^{1/\alpha_2 - 1/\alpha_1}
\]
is bounded over $t \in (0,1)$. This is possible only if $1/\alpha_2 - 1/\alpha_1 \ge 0$, so we obtain $\alpha_1 \ge \alpha_2$.

(ii) It is sufficient to note that $\mu_1 \lecx \mu_2$ if and only if $\mu_1 \leicx \mu_2$ and $m_1(\mu_1) = m_1(\mu_2)$.
The claim hence follows from (i).

(iii) Recall \cite[Thm~2.58]{Follmer_Schied_2004} that $\mu_1 \leicv \mu_2$ if and only if $H_1(t) \le H_2(t)$ for all $t \in [0,1]$.

Assume now that $\mu_1 \leicv \mu_2$. Then $H_1(1) \le H_2(1)$ implies $\lambda_1 \le \lambda_2$. Moreover, $H_1(0)=H_2(0)=0$ implies that $t^{-1}(H_1(t)-H_1(0)) \le t^{-1}(H_2(t)-H_2(0))$ for all $t \in (0,1)$, and by letting $t \to 0$, we get $c_1 = H_1'(0) \le H_2'(0) = c_2$, so that $c_1 \le c_2$. (Note that this reasoning indeed showed that $\mu_1 \lest \mu_2$ which automatically implies $\mu_1 \leicv \mu_2$.)

To prove the other direction of the claim, let us next assume that $\lambda_1 \le \lambda_2$ and $c_1 \le c_2$. Let us analyze separately the cases $\alpha_1 \ge \alpha_2$ and $\alpha_1 \le \alpha_2$.
\begin{enumerate}[(a)]
\item If $\alpha_1 \ge \alpha_2$, then
\[
 F_1^{-1}(s) = c_1 (1-s)^{-1/\alpha_1}
 \le c_2 (1-s)^{-1/\alpha_2} = F_2^{-1}(s)
\]
for all $s \in (0,1)$.
\item If $\alpha_1 \le \alpha_2$, then $\lambda_1 \le \lambda_2$ shows that
\[
 \frac{H_2(t)}{H_1(t)}
 = \left( \frac{\lambda_2}{\lambda_1} \right) \frac{1-(1-t)^{a_2}}{1-(1-t)^{a_1}}
 \ge \frac{1-(1-t)^{a_2}}{1-(1-t)^{a_1}},
\]
where $a_i = 1-1/\alpha_i$. Now $\alpha_1 \le \alpha_2$ implies $a_1 \le a_2$, Hence also $(1-t)^{a_1} \ge (1-t)^{a_2}$,
which shows that $H_2(t) \ge H_1(t)$ for all $t \in (0,1)$.
\end{enumerate}
\end{proof}

\bibliographystyle{alpha}
\bibliography{lslReferences}

\end{document}